\documentclass[11pt,a4paper]{article}

\usepackage{amsmath,amssymb,amsthm,bbm}

\usepackage[T1]{fontenc}
\usepackage[applemac]{inputenc}	
\usepackage[english]{babel}
\usepackage{geometry, color, graphicx}

\newtheorem{theorem}{Theorem}
\newtheorem{corollary}[theorem]{Corollary}
\newtheorem{lemma}[theorem]{Lemma}
\newtheorem{proposition}[theorem]{Proposition}

\newtheorem{definition}[theorem]{Definition} \theoremstyle{remark}

\newcommand{\comment}[1]{}

\newcommand{\R}{\mathbb R}

\newcommand{\la}{\lambda}

\newcommand{\de}{{\rm d}}
\newcommand{\supp}{{\rm supp}}
\newcommand{\Vol}{{\rm Vol}}

 \newgeometry{tmargin=2cm, bmargin=2cm, lmargin=2cm, rmargin=2cm}

\begin{document}

\large \title{On the entropy power inequality for the R\'enyi entropy of order $[0,1]$}

\author{Arnaud Marsiglietti\footnote{Supported in part by the Walter S. Baer and Jeri Weiss CMI Postdoctoral Fellowship. } \hspace{2mm} and James Melbourne\footnote{Supported by  NSF grant CNS 1544721.\newline \indent Parts of this paper were presented at the 2018 IEEE International Symposium on Information Theory, Vail, CO, USA, June 2018.}}

\date{}

\maketitle

\begin{abstract}

Using a sharp version of the reverse Young inequality, and a R\'enyi entropy comparison result due to Fradelizi, Madiman, and Wang (2016), the authors derive R\'enyi entropy power inequalities for log-concave random vectors when R\'enyi parameters belong to $[0,1]$. Furthermore, the estimates are shown to be sharp up to absolute constants. 

\end{abstract}


\section{Introduction}

Let $r \in [0,\infty]$. The R\'enyi entropy \cite{Ren61} of parameter $r$ is defined for continuous random vectors $X \sim f_X$ as
\begin{equation} \label{eq:definition of Renyi entropy}
    h_r(X) = \frac{1}{1-r} \log \left( \int_{\R^n} f_X^r(x) dx \right).
\end{equation}
We take the R\'enyi entropy power of $X$ to be
\begin{equation} \label{eq:definition of Renyi entropy POWER}
    N_r(X) = e^{\frac{2}{n} h_r(X)} = \left( \int_{\R^n} f_X^r(x) \, dx \right)^{\frac{2}{n} \frac{1}{1-r}}.
\end{equation}  

Three important cases are handled by continuous limits,
\begin{equation} \label{eq: 0 Renyi entropy power defn}
    N_0(X) = \Vol^{\frac 2 n}(\supp(X)), 
\end{equation}
\begin{equation} \label{eq: infinity renyi power defn}
    N_\infty(X) = \|f_X\|_\infty^{-2/n},
\end{equation}  
and $N_1(X)$ corresponds to the usual Shannon entropy power $N_1(X)= N(X) = e^{- \frac 2 n \int f \log f }$. Here, $\Vol(A)$ denotes the Lebesgue measure of a measurable set $A$, and $\supp(X)$ denotes the support of $X$.

The entropy power inequality (EPI) is the statement that Shannon entropy power of independent random vectors $X$ and $Y$ is super-additive 
\begin{equation} \label{eq:Shannon EPI}
	N(X+Y) \geq N(X) + N(Y).
\end{equation}
In this language we interpret the  Brunn-Minkowski inequality of Convex Geometry, classically stated as the fact that 
\begin{eqnarray} \label{eq: Brunn Minkowski}
\Vol(A+B) \geq \left(\Vol^{\frac 1 n}(A) + \Vol^{\frac 1 n}(B) \right)^n
\end{eqnarray}
for any pair of compact sets of $\mathbb{R}^n$ (see \cite{Gar02} for an introduction to the literature surrounding this inequality), as a R\'enyi-EPI corresponding to $r=0$.  That is, the Brunn-Minkowski inequality is equivalent to the fact that for independent random vectors $X$ and $Y$, the square root of the $0$-th R\'enyi entropy is super-additive,
\begin{equation} \label{eq:Brunn Minkowski as an REPI}
	N_0^{\frac 1 2}(X+Y) \geq N_0^{\frac 1 2}(X) + N_0^{\frac 1 2}(Y).
\end{equation}

The parallels between the two famed inequalities had been observed in the 1984 paper of Costa and Cover \cite{CC84}, and a unified proof using sharp Young's inequality was given in 1991 by Dembo, Cover, and Thomas \cite{DCT91}. Subsequently, analogs of further Shannon entropic inequalities and properties in Convex Geometry have been pursued.  For example the monotonicity of entropy in the central limit theorem (see \cite{ABBN04:1, MB07, TV06}),  motivated the investigation of quantifiable convexification of a general measurable set on repeated Minkowski summation with itself (see \cite{FMMZ16, FMMZ16:cras}). Motivated by Costa's EPI improvement \cite{Cos85b}, Costa and Cover conjectured that the volume of general sets when summed with a dilate of the Euclidean unit ball should have concave growth in the dilation parameter \cite{CC84}.   Though this was disproved for general sets in \cite{FM14}, open questions of this nature remain.

Conversely, V. Milman's reversal of the Brunn-Minkowski inequality (for symmetric convex bodies under certain volume preserving linear maps) \cite{Mil86} inspired Bobkov and Madiman to ask and answer whether the EPI could be reversed for log-concave random vectors under analogous mappings \cite{BM12:jfa}. In \cite{BM11:it} The authors also formulated an entropic version of Bourgain's slicing conjecture \cite{Bou86}, a longstanding open problem in convex geometry that has attracted a lot of attention.

A further example of an inequality at the interface of geometry and information theory can be found in \cite{BNT15}, where Ball, Nayar, and Tkocz conjectured the existence of an entropic Busemann's inequality \cite{Bus49}  for symmetric log-concave random variables and prove some partial results, see \cite{XMM16:isit} for an extension to ``$s$-concave'' random variables.

We refer to the survey \cite{MMX17:1} for further details on the connections between convex geometry and information theory.

Recently the super-additivity of more general R\'enyi functionals has seen significant activity, starting with Bobkov and Chistyakov \cite{BC14,BC15:1} where it is shown (the former focusing on $r=\infty$ the latter on $r \in (1,\infty)$) that for $r \in (1,\infty]$ there exist universal constants $c(r) \in (\frac 1 e, 1)$ such that for $X_i$ independent random vectors
\begin{equation} \label{eq:Bobkov Chistyakov REPI}
N_r(X_1 + \cdots + X_k) \geq c(r) \sum_{i=1}^k N_r(X_i).
\end{equation}
This was followed by Ram and Sason \cite{RS16} who used optimization techniques to sharpen bounds on the constant $c(r)$, which should more appropriately be written $c(r,k)$ as the authors were able to clarify the dependency on the number of summands as well as the R\'enyi parameter $r$.  Bobkov and Marsiglietti \cite{BM16} showed that for $r \in (1,\infty)$, there exists an $\alpha$ modification of the R\'enyi entropy power that preserved super-additivity.  More precisely taking $\alpha = \frac{r+1}{2}$, $r \in [1,\infty)$, and $X,Y$ independent random vectors 
\begin{equation} \label{eq:Bobkov Marsiglietti REPI}
	N_r^\alpha(X+Y) \geq N_r^\alpha(X) + N_r^\alpha(Y).
\end{equation}
This was sharpened by Li \cite{Li17} who optimized the argument of Bobkov and Marsiglietti.  The case of $r = \infty$ was studied using functional analytic tools by Madiman, Melbourne, and Xu \cite{MMX17:2, XMM17:isit:1} who showed that the $N_\infty$ functional enjoys an analog of the matrix generalizations of Brunn-Minkowski and the Shannon-Stam EPI due to Feder and Zamir \cite{ZF93,ZF98} and began investigation into discrete versions of the inequality in \cite{XMM17:isit:2}.

Conspicuously absent from the discussion above, and mentioned as an open problem in \cite{BC15:1,Li17,MMX17:1,RS16} are super-additivity properties of the R\'enyi entropy power when $r \in (0,1)$. In this paper, we address this problem, and provide a solution in the log-concave case (see Definition \ref{log-concave}). Our first main result is the following.

\begin{theorem}\label{main}
	
	Let $r \in (0,1)$. Let $X,Y$ be log-concave random vectors in $\R^n$. Then,
	\begin{eqnarray}\label{eq:main}
	N_r^{\alpha}(X+Y) \geq N_r(X)^{\alpha} + N_r(Y)^{\alpha},
	\end{eqnarray}
	where
	\begin{eqnarray}
		\alpha \triangleq \alpha(r)
		= 
		\frac{(1-r) \log 2 } { (1+r) \log(1+r) + r \log \frac 1  {4r}}.		\label{eq: alpha definition}
	\end{eqnarray} 
	
\end{theorem}

Furthermore, and in contrast to some previous optimism (see, e.g., \cite{Li17}), these estimates are somewhat sharp for log-concave random vectors. Indeed, letting $\alpha_{opt}=\alpha_{opt}(r)$ denote the infimum over all $\alpha$ satisfying the inequality \eqref{eq:main} for log-concave random vectors, we have
\begin{equation} \label{eq:lower bounds on alpha optimal}
	\max \left\{ 1, \frac{(1-r) \log 2}{2 \log \Gamma(1+r) + 2 r \log \frac 1 r} \right \}  \leq \alpha_{opt} \leq \alpha(r),
\end{equation}
(see Proposition \ref{lower-bound} in Section \ref{sec:lower}). Unsurprisingly, the bounds \eqref{eq: alpha definition} and \eqref{eq:lower bounds on alpha optimal} imply that
\begin{equation} \lim_{r \to 1} \alpha(r) = \lim_{r \to 1} \alpha_{opt}(r) = 1, \end{equation}
recovering the usual EPI. In fact the ratio of the lower and upper bounds satisfies
\begin{equation} \label{eq: ratio of upper and lower bounds}
    \frac{2 \log \Gamma(1+r) + 2 r \log \frac 1 r }{(1+r) \log(1+r) + r \log \frac 1  {4r}} \to \frac 1 2
\end{equation}
with $r \to 0$ as can be seen by applying L'H\^opital's rule and the strict convexity of $\log \Gamma (1+r)$.  It can be verified numerically that the derivative of \eqref{eq: ratio of upper and lower bounds} is strictly positive on $(0,1)$.  Thus the $\alpha(r)$ derived  cannot be improved beyond a factor of $2$. 

More strikingly, as $r \to 0$ the bounds derived force both $\alpha_{opt}$ and $\alpha$ to be of the order
$(- r \log r)^{-1}$.
Thus, $\alpha_{opt}(r) \to + \infty$ for $r \to 0$, while $\alpha_{opt}(0) =1/2$ by the Brunn-Minkowski inequality. 
Nevertheless, in the case that the random vectors are uniformly distributed we do have better behavior.

\begin{theorem}\label{main2}
	Let $r \in (0,1)$. Let $X,Y$ be uniformly distributed random vectors on compact sets. Then,
	\begin{eqnarray} \label{eq: REPI for Uniformly distributed random vectors}
		N_r^{\beta}(X+Y) \geq N_r^{\beta}(X) + N_r^{\beta}(Y),
	\end{eqnarray}
	where
	\begin{eqnarray}
		\beta \triangleq \beta(r)
		=
		\frac{(1-r)\log 2 }{2 \log 2 + r \log r - (r+1) \log (r+1)}. \label{eq: definition of beta}
	\end{eqnarray}
\end{theorem}

Stated geometrically, Theorem \ref{main2} is the following generalization of the Brunn-Minkowski inequality.

\begin{theorem}\label{geom-equiv}
	Let $r \in (0,1)$. Let $A,B$ be compact sets in $\mathbb{R}^n$. Then, letting $X$ and $Y$ denote independent random vectors distributed uniformly on the respective sets $A$ and $B$,
	\begin{eqnarray}          e^{h_r(X+Y)} \geq \left(\Vol^{\gamma}(A) + \Vol^{\gamma}(B) \right)^{\frac 1 \gamma} 
	\end{eqnarray}
	where
	$\gamma \triangleq 2\beta/n.$
\end{theorem}

Theorems \ref{main2} and \ref{geom-equiv} can be understood as a family of R\'enyi-EPIs for uniform distributions interpolating between the Brunn-Minkowski inequality and EPI. Indeed $\lim_{r\to 0} \gamma = 1/n$, while $ e^{h_r(X+Y)}$ increases to  $\Vol(A+B)$, and we recover the Brunn-Minkowski inequality \eqref{eq: Brunn Minkowski}.  Observe, $\lim_{r \to 1} \beta =1$ gives the usual EPI in the special case that the random vectors are uniform distributions.   Note also that the exponent $\beta$ in \eqref{eq: definition of beta} is identical to the exponent obtained in \cite[Theorem 2.2]{Li17} for $r>1$.

We also approach the R\'enyi EPI of the form \eqref{eq:Bobkov Chistyakov REPI} and obtain the following result.

\begin{theorem} \label{thm: bobkov chistyakov Repi for small r}
    
    Let $r \in (0,1)$. For all independent log-concave random vectors $X_1, \dots, X_k$ in $\R^n$,
    \begin{eqnarray}\label{eq:k-tuples}
        N_r(X_1 + \cdots + X_k) \geq c(r,k) \sum_{i=1}^k N_r(X_i)
    \end{eqnarray}
    where 
    \begin{eqnarray}
        c(r,k) \geq r^{\frac 1 {1-r}} \left( 1 + \frac 1 {k |r'|} \right)^{1 + k |r'|}.
    \end{eqnarray}
\end{theorem}
This bound is shown to be tight up to absolute constants as well. Indeed, we will see in Proposition \ref{sharp-k-tuples} in Section \ref{sec:k-tuples} that the largest constant $c_{opt}(r)$ satisfying
    \begin{eqnarray}
        N_r(X_1 + \cdots + X_k) \geq c_{opt}(r) \sum_{i=1}^k N_r(X_i)
    \end{eqnarray}
    for any $k$-tuples of independent log-concave random vectors satisfies
    \begin{eqnarray}
       e r^{\frac 1 {1-r}} \leq c_{opt}(r) \leq \pi r^{\frac 1 {1-r}}.
    \end{eqnarray}

\section{Preliminaries}

For $p \in [0,\infty]$, we denote by $p'$ the conjugate of $p$,
\begin{equation}
	\frac 1 {p} + \frac 1 {p'} = 1.
\end{equation}
For a non-negative function $f \colon \mathbb{R}^n \to [0, +\infty)$ we introduce the notation
\begin{equation}
	\|f\|_p = \left( \int_{\mathbb{R}^n} f^p(x) dx \right)^{1/p}.
\end{equation}

\begin{definition}\label{log-concave}
	A random vector $X$ in $\R^n$ is log-concave if it possesses a log-concave density $f_X \colon \R^n \to [0,+\infty)$ with respect to Lebesgue measure.  In other words, for all $\lambda \in (0,1)$ and $x,y \in \mathbb{R}^n$,
	\begin{equation}
		f_X((1-\la)x + \la y) \geq f_X^{1-\la}(x)f_X^{\la}(y).
	\end{equation}
	Equivalently $f_X$ can be written in the form $e^{-V}$, where $V$ is a proper convex function.
\end{definition}

Log-concave random vectors and functions are important classes in many disciplines. In the context of information theory, several nice properties involving entropy of log-concave random vectors were recently established (see, e.g., \cite{BN12, BM11:it, CFP16, MK17, Tos15:2, Tos15:1}). Significant examples are Gaussian and exponential distributions as well as any uniform distribution on a convex set.  

The main tool in establishing Theorems \ref{main}, \ref{main2} and \ref{thm: bobkov chistyakov Repi for small r} is the reverse form of the sharp Young inequality. The reversal of Young's inequality for parameters in $[0,1]$ is due to Leindler \cite{Lei72b}, while sharp constants were obtained independently by Beckner, and Brascamp and Lieb:

\begin{theorem}[\cite{Bec75,BL76b}]

Let $0 \leq p,q,r \leq 1$ such that $\frac{1}{p'} + \frac{1}{q'} = \frac{1}{r'}$. Then,
\begin{eqnarray}\label{young}
\|f \star g\|_r \geq C^{\frac{n}{2}} \|f\|_p \|g\|_q,
\end{eqnarray}
where
\begin{equation} \label{eq: c_m defintion} C = C(p,q,r) = \frac{c_p c_q}{c_r}, \qquad c_m = \frac{m^{1/m}}{|m'|^{1/m'}}. \end{equation}

\end{theorem}

Let us recall the information-theoretic interpretation of Young's inequality.  Given independent random vectors $X$ with density $f$ and $Y$ with density $g$, the random vector $X+Y$ will be distributed according to $f \star g$. Observe that the  $L_p$ ``norms''  have the following expression as R\'enyi entropy powers, \begin{equation} \label{eq: Lp norms to Renyi conversion}
\|f\|_r = N_r(X)^{-\frac{n}{2 r'}} = N_r(X)^{\frac{n}{2 |r'|}}.
\end{equation}
Hence, we can rewrite \eqref{young} as follows,
\begin{eqnarray}\label{young-info}
N_r(X+Y)^{\frac{1}{|r'|}} \geq C N_p(X)^{\frac{1}{|p'|}} N_q(Y)^{\frac{1}{|q'|}}.
\end{eqnarray}
This is an information-theoretic interpretation of the sharp Young inequality, which was developed in \cite{DCT91}.

We also need a R\'enyi comparison result for log-concave random vectors that the authors first learned from private communication \cite{MokshayPrivateCom}.  Though the result is implicit in \cite{FMW16}, we give a derivation in the appendix for the convenience of the reader.  A generalization of the result to $s$-concave random variables (see \cite{bobkov2016hyperbolic,Borell74}) is planned to be included in a revised version of \cite{FLM15}.

\begin{lemma}[Fradelizi-Madiman-Wang \cite{FMW16,MokshayPrivateCom}]\label{reverse}

Let $0<p<q$. Then, for every log-concave random vector $X$,
\begin{equation} N_q(X) \leq N_p(X) \leq \frac{p^{\frac{2}{p-1}}}{q^{\frac{2}{q-1}}} N_q(X). \end{equation}

\end{lemma}

The first inequality is classical and holds for general $X$, and follows from the expression $N_p(X)^{n/2} = \left(\mathbb{E}f^{p-1}(X)\right)^{-1/(p-1)}$.  Indeed, the increasingness of the function $s \mapsto (\mathbb{E} Y^s)^{1/s}$ for a positive random variable $Y$ and $s \in (-\infty,\infty)$, which follows from Jensen's inequality, implies the decreasingness of R\'enyi entropy powers.  The content of Fradelizi, Madiman, and Wang's result is thus the second inequality, that this decrease is not too rapid for log-concave random vectors.

We will also have use for a somewhat technical but elementary Calculus result.

\begin{lemma} \label{lem:Calculus}
	Let $c>0$. Let $L,F \colon [0, c] \to [0,\infty)$ be twice differentiable on $(0,c]$, continuous on $[0,c]$, such that $L(0) = F(0)= 0$ and $L'(c)=F'(c) =0$.  Let us also assume that $F(x)>0$ for $x>0$, that $F$ is strictly increasing, and that $F'$ is strictly decreasing. Then 
	$\frac{L''}{F''}$ increasing on $(0,c)$ implies that $\frac{L}{F}$ is increasing on $(0,c)$ as well. In particular,
	\begin{equation}
		\max_{x \in [0,c]} \frac{L(x)}{F(x)} = \frac{L(c)}{F(c)}.
	\end{equation}
\end{lemma}

The proof is an exercise in  Cauchy's mean value theorem.

\begin{proof}
For $0< u<v< c$, by Cauchy's mean value theorem
	\begin{eqnarray}
	    \frac{L'({c}) - L'(v)}{F'({c}) - F'(v)} &=& \frac{L''(c_1)}{F''(c_1)}, \label{eq:Cauchy MVT 1}
	    \\
		\frac{L'(v) - L'(u)}{F'(v) - F'(u)} &=& \frac{L''(c_0)}{F''(c_0)}, \label{eq:Cauchy MVT 2}
	\end{eqnarray} 
	for some $c_0 \in (u,v)$ and $c_1 \in (v, c)$. 
	Thus,
	\begin{eqnarray}
        \frac{ L'(v) }{ F'(v) }  
            &=& 
                \frac{ L'(c) - L'(v) }{F'(c) - F'(v) } \label{1} 
                    \\
            &= &
                \frac{ L''(c_1) }{ F''(c_1) } \label{2} 
                    \\
            &\geq &
                \frac{ L''(c_0) }{ F''(c_0) } \label{3} 
                    \\
            &= &
                \frac{ L'(v) - L'(u) }{ F'(v) - F'(u) } \label{4}
    \end{eqnarray}
where \eqref{1} holds by the assumption that $L'(c) = F'(c) = 0$, \eqref{2} and \eqref{4} follow from \eqref{eq:Cauchy MVT 1} and \eqref{eq:Cauchy MVT 2} respectively, and \eqref{3} holds by the assumption that $\frac{L''}{F''}$ is monotonically increasing in $(0,c)$. The inequality
	\begin{equation}
	\frac{L'(v)}{F'(v)} \geq \frac{L'(v) - L'(u)}{F'(v) - F'(u)},
	\end{equation}
	is equivalent to
	\begin{equation}
	-L'(v)F'(u) \leq -L'(u) F'(v),
	\end{equation}
	because $F'$ is non-negative and strictly decreasing on $(0,c)$.  Thus $L'(v)/F'(v) \geq L'(u)/F'(u)$ since $F' \geq 0$. That is, $L'/F'$ is non-decreasing on $(0,c)$.  Now we can apply a similar argument to show that $L/F$ is non-decreasing. Again Cauchy's mean value theorem, for $0 < u < v < c$ we have
	\begin{eqnarray}
		\frac{L(u) - L(0)}{F(u) - F(0)} &=& \frac{L'(c_0)}{F'(c_0)},
		\\
		\frac{L(v) - L(u)}{F(v) - F(u)} &=& \frac{L'(c_1)}{F'(c_1)},
	\end{eqnarray} 
	for some $c_0 \in (0,u)$ and $c_1 \in (u, v)$. Thus by the proven non-decreasingness of $\frac{L'}{F'}$ and the fact that $F(0) = L(0) = 0$ the above implies
	\begin{equation}
	\frac{L(v) - L(u)}{F(v) - F(u)} \geq  \frac{L(u)}{F(u)}.
	\end{equation}
	Since $F$ is non-negative and strictly increasing on $(0,c)$, we have
	\begin{equation}
	L(v) F(u) \geq L(u) F(v).
	\end{equation}
	Thus it follows that $L/F$ is indeed non-decreasing.
\end{proof}

\section{Proof of Theorem \ref{main}}\label{th1}

We first combine Lemma \ref{reverse} and the information-theoretic formulation of reverse Young's inequality \eqref{young-info}. Observe that for $p,q,r \in (0,1)$ satisfying the equation $\frac 1 {p'} + \frac 1 {q'} = \frac 1 {r'}$ forces $p,q > r$. Thus, our invocation of Lemma \ref{reverse} is necessarily at the expense of the two constants below,
\begin{eqnarray}\label{eq:1}
N_r(X+Y)^{\frac{1}{|r'|}} \geq C \left( \frac{p^{\frac{2}{p-1}}}{r^{\frac{2}{r-1}}} \right)^{\frac{1}{|p'|}}  \left( \frac{q^{\frac{2}{q-1}}}{r^{\frac{2}{r-1}}} \right)^{\frac{1}{|q'|}} N_r(X)^{\frac{1}{|p'|}} N_r(Y)^{\frac{1}{|q'|}}.
\end{eqnarray}
Since
\begin{equation} \frac{1}{|p'|} = -\frac{1}{p'} = \frac{1}{p} - 1 = \frac{1-p}{p}, \end{equation}

we deduce that
\begin{equation} \frac{1}{|p'|(p-1)} = -\frac{1}{p}. \end{equation}
Also, we have
\begin{equation} \frac{1}{|p'|} + \frac{1}{|q'|} = \frac{1}{|r'|}. \end{equation}
Hence, we can rewrite \eqref{eq:1} as,
\begin{eqnarray}
	N_r(X+Y)^{\frac{1}{|r'|}} \geq C \,\, \frac{p^{-\frac{2}{p}} q^{-\frac{2}{q}}}{r^{-\frac{2}{r}}} N_r(X)^{\frac{1}{|p'|}} N_r(Y)^{\frac{1}{|q'|}} = A(p,q,r) N_r(X)^{\frac{1}{|p'|}} N_r(Y)^{\frac{1}{|q'|}},
\end{eqnarray}
where
\begin{equation} \label{eq: defintion of constant A} A(p,q,r) = \frac{c_p c_q}{c_r} \frac{r^{\frac{2}{r}}}{p^{\frac{2}{p}} q^{\frac{2}{q}}}. \end{equation}
Equivalently,
\begin{eqnarray} \label{eq:Multiplicative EPI}
N_r(X+Y) \geq A(p,q,r)^{|r'|} N_r(X)^{\frac{|r'|}{|p'|}} N_r(Y)^{\frac{|r'|}{|q'|}}.
\end{eqnarray}

We collect these arguments to state the following result, actually stronger than Theorem \ref{main}.

\begin{theorem}
    Let $r \in (0,1)$. Let $X,Y$ be independent log-concave vectors in $\mathbb{R}^n$. For $0<p,q<1$ satisfying $\frac 1 {p'} + \frac 1 {q'} = \frac 1 {r'}$, we have
    \begin{equation}
        N_r(X+Y) \geq A(p,q,r)^{|r'|} N_r(X)^{\frac{|r'|}{|p'|}} N_r(Y)^{\frac{|r'|}{|q'|}}
    \end{equation}
    with $A(p,q,r)$ as defined in \eqref{eq: defintion of constant A}.
\end{theorem}
Thus to complete our proof of Theorem \ref{main} it suffices to obtain for a fixed $r \in (0,1)$, an $\alpha>0$ such that for any given pair of independent log-concave random vectors $X$ and $Y$, there exist $0\leq p,q \leq 1$ such that $\frac 1 {p'} + \frac 1 {q'} = \frac 1 {r'}$ and
\begin{eqnarray} \label{eq: reduction to algebra}
	 A(p,q,r)^{\alpha|r'|} N_r(X)^{\frac{\alpha|r'|}{|p'|}} N_r(Y)^{\frac{\alpha|r'|}{|q'|}} \geq N_r^\alpha(X) + N_r^\alpha(Y).
\end{eqnarray}

Let us observe that there is nothing probabilistic about equation \eqref{eq: reduction to algebra}. If we write
$x = N_r(X)^{\alpha}$, $y = N_r(Y)^{\alpha}$, our R\'enyi-EPI is implied by the following algebraic inequality.

\begin{proposition}
	Given $r \in (0,1)$ and taking
	\begin{eqnarray}\label{alpha}
	\alpha = \frac{(1-r) \log 2 } { (1+r) \log(1+r) + r \log \frac 1 {4r}},
	\end{eqnarray}
	then for any $x,y >0$ there exist $0 < p,q < 1$ satisfying  $\frac 1 {p'} + \frac 1 {q'} = \frac 1 {r'}$ such that
	\begin{eqnarray} \label{eq:alg lem}
		 A(p,q,r)^{\alpha|r'|} x^{\frac{|r'|}{|p'|}} y^{\frac{|r'|}{|q'|}} \geq x + y.	
	\end{eqnarray}
\end{proposition} 

\begin{proof}
Using the homogeneity of equation \eqref{eq:alg lem}, we may assume without loss of generality that
\begin{equation} x+y = \frac{1}{|r'|}. \end{equation}
We then choose admissible $p,q$ by selecting $\frac{1}{p'} = -x$ and $\frac{1}{q'} = -y$. Hence, equation \eqref{eq:alg lem} becomes
\begin{equation} A(p,q,r)^{\alpha} \geq \frac{(x+y)^{x+y}}{x^x y^y}. \end{equation}
Let us note that $A(p,q,r) \geq 1$ (we will prove this fact in the appendix based on the description of $A(p,q,r)$ in \eqref{eq: x y description of A}), so that taking logarithms we can choose
\begin{equation} \alpha = \sup \frac{\log \left( \frac{(x+y)^{x+y}}{x^x y^y} \right) }{\log(A(p,q,r))}, \end{equation}
where the sup runs over all $x,y > 0$ satisfying $x+y = \frac{1}{|r'|}$ (recall that $r \in (0,1)$ is fixed). We claim that this is exactly the $\alpha$ defined in \eqref{alpha}. To establish this fact, let us first rewrite $A(p,q,r)$ in terms of $x$ and $y$. From,
\begin{equation} p = \frac{1}{x+1}, \qquad q = \frac{1}{y+1}, \qquad r = \frac{1}{x+y+1}, \label{eq: p,q,r definitions} \end{equation}
we can write,
\begin{eqnarray}
c_p & = & \frac{p^{1/p}}{|p'|^{1/p'}} = \frac{\frac{1}{(x+1)^{x+1}}}{\frac{1}{x^{-x}}} = \frac{1}{x^x (x+1)^{x+1}}, \label{eq: c_p definition} \\ c_q & = & \frac{1}{y^y (y+1)^{y+1}}, \label{eq: c_q defintion} \\ c_r & = & \frac{1}{(x+y)^{x+y} (x+y+1)^{x+y+1}} \label{eq: c_r defintion}.
\end{eqnarray}
From \eqref{eq: c_p definition} - \eqref{eq: c_r defintion} it follows that,
\begin{eqnarray}
A(p,q,r) &= & \frac{c_p c_q}{c_r} \frac{r^{\frac{2}{r}}}{p^{\frac{2}{p}} q^{\frac{2}{q}}} 
\\ 
    & = & 
        \frac{(x+y)^{x+y} (x+1)^{x+1} (y+1)^{y+1}}{x^x y^y (x+y+1)^{(x+y+1)}}. \label{eq: x y description of A}
\end{eqnarray}
Let us denote
\begin{equation} F(x) \triangleq \log \left( \frac{(x+y)^{x+y}}{x^x y^y} \right) = \frac{1}{|r'|} \log \left( \frac{1}{|r'|} \right) -x \log(x) - \left( \frac{1}{|r'|} - x \right) \log \left( \frac{1}{|r'|} - x \right), \end{equation}
and
\begin{equation} G(x) \triangleq \log \left( \frac{(x+y)^{x+y} (x+1)^{x+1} (y+1)^{y+1}}{x^x y^y (x+y+1)^{(x+y+1)}} \right) = F(x) - L(x), \end{equation}
where
\begin{equation} L(x) \triangleq \left( \frac{1}{|r'|} + 1 \right) \log \left( \frac{1}{|r'|} + 1 \right) - (x+1)\log(x+1) - \left(\frac{1}{|r'|} - x + 1 \right) \log \left(\frac{1}{|r'|} - x + 1 \right). \end{equation}
Our claim is that
\begin{equation}
	\alpha = \sup \frac{F}{G} = \sup \frac{F}{F-L} = \left( 1- \sup \frac{L}{F} \right)^{-1}.
\end{equation}
We invoke Lemma \ref{lem:Calculus}, to prove that the ratio $L/F$ is increasing on $[0,1/2|r'|]$.  Indeed, taking derivatives it is easy to see that $F$ is positive and increasing on $(0, 1/2|r'|]$, and its derivative $F'$ is strictly decreasing on the same interval.  
Furthermore, $\frac{L''}{F''}$ is non-decreasing on $(0, \frac{1}{2|r'|})$. Indeed,
\begin{equation} \frac{L''(x)}{F''(x)} = \frac{\frac{1}{|r'|}+2}{\frac{1}{|r'|}} \frac{x (\frac{1}{|r'|}-x)}{(x+1)(\frac{1}{|r'|}-x+1)}, \end{equation}
and one can see that this is non-decreasing when $x \in (0, \frac{1}{2|r'|})$ again, by taking the derivative. Now by Lemma \ref{lem:Calculus} applied to $F, L$, and $c = \frac 1 {2 |r'|}$ we have
\begin{equation}
	\sup \left( 1- \frac{L(x)}{F(x)} \right)^{-1} = \left( 1- \frac{L(c)}{F(c)} \right)^{-1} = \left( 1- \frac{L(1/2|r'|)}{F(1/2|r'|)} \right)^{-1}.
\end{equation}
Let us compute $F(c)$ and $L(c)$, with $c = \frac 1 {2 |r'|}$. We have
	\begin{eqnarray}
		F(c) 
			=
				2c \log 2c - 2c \log (c) = \frac{(1-r) \log 2}{r},
	\end{eqnarray}
and
	\begin{eqnarray}
		L(c)
			&=&
				(2c+1) \log(2c+1) - 2(c+1) \log(c+1)
					\\
			&=&
				\frac{\log \left(\frac{2}{ 1+r } \right) }{r} - \log \left( \frac{r+ 1 }{2r}  \right).
	\end{eqnarray}
	Thus
	\begin{eqnarray}
		 \alpha 
			 &=&
				 \left(1- \frac{L(c)}{F(c)}	\right)^{-1}
					 \\
			&=&
				\left( 1 - \frac{	\frac{\log \left(\frac{2}{ 1+r } \right) }{r} - \log \left( \frac{r+ 1 }{2r}  \right) }{	\frac{(1-r) \log 2}{r} } \right)^{-1}
					\\
			&=&
				\frac{(1-r) \log 2 } { (1+r) \log(1+r) + r \log \frac 1  {4r}}.
	\end{eqnarray}
\end{proof}




\section{Proof of Theorem \ref{main2}}

The proof is very similar to the proof of Theorem \ref{main}. The improvement is by virtue of the fact that for $U$ a random vector uniformly distributed on a set $A \subset \mathbb{R}^n$, the R\'enyi entropy is determined entirely by the volume of $A$, and is thus independent of parameter. Indeed,
\begin{eqnarray}
N_r(U) = \left(\int_{\mathbb{R}^n} \left(\mathbbm{1}_A (x)/\Vol(A) \right)^r dx \right)^{2/n(1-r)}  = \Vol(A)^{2/n}.
\end{eqnarray} 
We again use the information-theoretic version of the sharp Young inequality (see \eqref{young-info}):
\begin{equation} N_r(X+Y)^{\frac{1}{|r'|}} \geq C N_p(X)^{\frac{1}{|p'|}} N_q(Y)^{\frac{1}{|q'|}}. \end{equation}
Now, since $X$ and $Y$ are uniformly distributed, we have
\begin{equation} N_p(X) = N_r(X),  \qquad N_q(Y) = N_r(Y). \end{equation}
Hence,
\begin{eqnarray}\label{unif}
N_r(X+Y) \geq C^{|r'|} N_r(X)^{\frac{|r'|}{|p'|}} N_r(Y)^{\frac{|r'|}{|q'|}}.
\end{eqnarray}
Let us raise \eqref{unif} to the power $\beta$, and put $x=N_r(X)^{\beta}$, $y=N_r(Y)^{\beta}$. As before, we can assume that $x+y = \frac{1}{|r'|}$. Thus, it is enough to show that
\begin{equation} C^{\beta |r'|} x^{\frac{|r'|}{|p'|}} y^{\frac{|r'|}{|q'|}} \geq \frac{1}{|r'|}, \end{equation}
for some admissible $(p,q)$. Let us choose $p,q$ such that $x=\frac{1}{|p'|}$ and $y=\frac{1}{|q'|}$. The inequality is valid since
\begin{equation} \beta = \sup \frac{\log \left( \frac{(x+y)^{x+y}}{x^x y^y} \right) }{\log(C)} = \sup \frac{\log \left( \frac{(x+y)^{x+y}}{x^x y^y} \right) }{\log \left( \frac{(x+y)^{x+y}}{x^x y^y} \frac{(x+y+1)^{x+y+1}}{(x+1)^{x+1} (y+1)^{y+1}} \right)}, \end{equation}
where the sup runs over all $x,y > 0$ satisfying $x+y = \frac{1}{|r'|}$ (recall that $r \in (0,1)$ is fixed). Indeed, as in Section \ref{th1}, it is a consequence of Lemma \ref{lem:Calculus} that the sup is attained at $x = \frac{1}{2 |r'|}$ and from this the result follows.

\section{Lower bound on the optimal exponent}\label{sec:lower}

\begin{proposition}\label{lower-bound}

The optimal exponent $\alpha_{opt}$ that satisfies \eqref{eq:main} verifies,
\begin{equation} \max \left\{1, \frac{(1-r) \log 2}{2 \log \Gamma(r+1) + 2 r \log \frac 1 r} \right\} \leq \alpha_{opt} \leq \frac{(1-r) \log 2 } { (1+r) \log(1+r) + r \log \frac 1 {4r}}	\label{eq: upper and lower on alpha opt}. \end{equation}

\end{proposition}

Let us remark that smooth interpolation of Brunn-Minkowski and the EPI as in Theorem \ref{main2}, cannot hold for any class of random variables that contains the Gaussians. Indeed, let $Z_1$ and $Z_2$ be i.i.d. standard Gaussians. Hence, $Z_1 + Z_2 \sim \sqrt{2} Z_1$, and by homogeneity of R\'enyi entropy,
	\begin{equation}
	N_r^{\alpha}(Z_1+Z_2) = 2^{\alpha} N_r^\alpha(Z_1),
	\end{equation}
	while
	\begin{equation}
	N_r^{\alpha}(Z_1) + N_r^{\alpha}(Z_2) = 2 N_r^{\alpha}(Z_1).
	\end{equation}
	It follows that for a modified R\'enyi-EPI to hold, even when restricted to the class of log-concave random vectors, we must have $2^{\alpha} \geq 2$.  That is, $\alpha \geq 1$.
	
	We now show by direct computation on the exponential distribution on $(0,\infty)$ the lower bounds on $\alpha_{opt}$.
	
	Let $X \sim f_X$ be a random variable with exponential distribution, $f_X(x) = \mathbbm{1}_{(0,\infty)}(x) e^{-x}$. The computation of the R\'enyi entropy of $X$ is an obvious change of variables,
	\begin{eqnarray}
	N_r(X) = \left( \int f_X^r \right)^{\frac 2 {1-r}} = \left( \int_0^\infty e^{-rx} dx \right)^{\frac 2 {1-r}} = \left(	\frac 1 r \right)^{\frac 2 {1-r}}.
	\end{eqnarray}	
	Let $Y$ be an independent copy of $X$. The density of $X+Y$ is
	\begin{eqnarray}
	f*f(x) 
	&=& \int_{-\infty}^{\infty} \mathbbm{1}_{(0,\infty)}(x-y) e^{-(x-y)} \mathbbm{1}_{(0,\infty)}(y) e^{-y} dy
	\\
	&=&
	\mathbbm{1}_{(0,\infty)} x e^{-x}.
	\end{eqnarray}
	Hence,
	\begin{eqnarray}
	N_r(X+Y)
	&=&
	\left( \int \mathbbm{1}_{(0,\infty)}(x) x^r e^{-rx} dx \right)^{\frac 2 {1-r}}
	\\
	&=&
	\left( \frac {\Gamma(r+1)} {r^{r+1}}   \right)^{\frac 2 {1-r}}.
	\end{eqnarray}
	Since the optimal exponent $\alpha_{opt}$ satisfies
	\begin{equation}
	N_r^{\alpha_{opt}}(X+Y) \geq 2 N_r^{\alpha_{opt}}(X),
	\end{equation}
	we have
	\begin{equation}
	\left( \frac {\Gamma(r+1)} {r^{r+1}}   \right)^{\frac {2\alpha_{opt}} {1-r}} 
	\geq 2 \left(	\frac 1 r \right)^{\frac {2\alpha_{opt}} {1-r}}.
	\end{equation}
	Canceling and taking logarithms, this rearranges to
	\begin{equation}
	\log \Gamma(r+1) + r \log \frac 1 r \geq \frac{(1-r) \log 2}{2\alpha_{opt}},
	\end{equation}
	which implies that we must have
	\begin{equation}
		\alpha_{opt} \geq \frac{(1-r) \log 2 }{ 2 ( \log \Gamma(r+1) + r \log \frac 1 r)}.
	\end{equation}
	Note that by the log-convexity of $\Gamma$ and the fact that $\Gamma(1) = \Gamma(2) =1$, we have $\log(\Gamma(1+r)) \leq 0$, which implies
	\begin{equation}
	\alpha_{opt} \geq \frac{(1-r) \log 2 }{2 r \log \frac 1 r}.
	\end{equation}
	In particular we must have $\alpha_{opt} r^{1-\varepsilon} \to \infty$ with $r \to 0$, for any $\varepsilon>0$.

\section{Proof of Theorem \ref{thm: bobkov chistyakov Repi for small r}}\label{sec:k-tuples}

The reverse sharp Young inequality can be generalized to $k \geq 2$ functions in the following way.

\begin{theorem}[\cite{BL76b}] \label{thm: sharp young multiple functions}
    Let $f_1, \dots, f_k \colon \mathbb{R}^n \to \mathbb{R}$ and $r, r_1, \dots, r_k \in (0,1)$ such that $\frac 1 {r_1'} + \cdots + \frac 1 {r_k'} = \frac 1 {r'}$. Then,
    \begin{equation}
        \| f_1 * \cdots * f_k \|_r \geq C^{\frac n 2} \prod_{i=1}^k \| f_i \|_{r_i}.
    \end{equation}
    Here,
    \begin{equation} \label{eq:extendend young constants}
        C = C(r,r_1, \dots, r_k) = \frac{\prod_{i=1}^k c_{r_i}}{c_r},
    \end{equation}
    where we recall that $c_m$ is defined in \eqref{eq: c_m defintion} as $c_m  = \frac{m^{\frac 1 m}}{|m'|^{\frac 1 {m'}}}$.
\end{theorem}
We have the following information-theoretic reformulation for $X_1, \dots, X_k$ independent random vectors,
\begin{equation}\label{eq:gen-young}
    N_r(X_1 + \cdots + X_k) \geq C^{|r'|} \prod_{i=1}^k N_{r_i}^{|r'|/|r_i'|}(X_i).
\end{equation}
Thus when we restrict to log-concave random vectors $X_i$, $1 \leq i \leq k$, and invoke Lemma \ref{reverse}, we can collect our observations as the following.

\begin{theorem} \label{thm: General Young + log-conncave comparison}
Let $r, r_1, \dots, r_k \in (0,1)$ such that $\sum_{i=1}^k \frac 1 {r_i'} = \frac  1 {r'}$. Let $X_1, \dots, X_k$ be independent log-concave random vectors. Then,
\begin{eqnarray}\label{thm:eq-gen}
    N_r(X_1 + \cdots + X_k) 
        &\geq& 
            A^{|r'|} \prod_{i=1}^k N_r^{t_i}(X_i),
\end{eqnarray}
where $t_i = r'/r_i'$ and $A = A(r,r_1, \dots, r_n) = \frac{ \prod_{i=1}^k A_{r_i}}{A_r}$ with $A_m = \frac{ |m'|^{\frac 1 {|m'|}}}{m^{\frac 1 m}}$.
\end{theorem}

\begin{proof}
By combining \eqref{eq:gen-young} with Lemma \ref{reverse}, we obtain
\begin{eqnarray}
    N_r(X_1 + \cdots + X_k) 
        &\geq& 
            C^{|r'|} \prod_{i=1}^k \left( N_r(X_i) \left(\frac{ r^{|r'|/r}}{r_i^{|r_i'|/r_i}} \right)^2 \right)^{|r'|/|r_i'|}
                \\
        &=&
            C^{|r'|} \prod_{i=1}^k \left(\frac{ r^{|r'|/r}}{r_i^{|r_i'|/r_i}} \right)^{\frac{2|r'|}{|r_i'|}} \prod_{i=1}^k N_r^{\frac{r'}{r'_i}}(X_i)
                \\
        &=&
            A(r,r_1, \dots, r_k)^{|r'|} \prod_{i=1}^k N_r^{t_i}(X_i).
\end{eqnarray}
\end{proof}

Now let us show that Theorem \ref{thm: General Young + log-conncave comparison} implies a super-additivity property for the R\'enyi entropy and independent log-concave vectors.

\begin{proof}[Proof of Theorem \ref{thm: bobkov chistyakov Repi for small r}]
By the homogeneity of equation \eqref{eq:k-tuples}, we can assume without loss of generality that $\sum_{i=1}^k N_r(X_i) = 1$.  From Theorem \ref{thm: General Young + log-conncave comparison}, for every $r_1, \dots, r_k \in (0,1)$ such that $\sum_i \frac 1 {r_i'} = \frac 1 {r'}$ we have
\begin{eqnarray}
   N_r(X_1 + \cdots + X_k) 
        &\geq& 
            A^{|r'|} \prod_{i=1}^k N_r^{t_i}(X_i)
                \\
        &=&
            \frac{ r^{\frac{|r'|} r}\prod_{i=1}^k \left(\frac{ |r_i'|^{\frac 1 {|r_i'|}}}{r_i^{\frac 1 {r_i}}}\right)^{|r'|}N_r^{t_i}(X_i)}{ |r'|},
\end{eqnarray}
where $t_i = r'/r_i'$. Thus,
\begin{eqnarray}
    \log N_r(X_1 + \cdots + X_k)
        &\geq&
            \sum_{i=1}^k\left( t_i \log \frac{|r'|}{t_i}  - |r'| \frac {\log r_i}{r_i}  \right) 
                \\ \nonumber
        && +  \left(  \frac {|r'| \log{r}} r  - {\log {|r'|}} \right)
            \\ \nonumber
        && +
            \sum_{i=1}^k t_i \log N_r(X_i).
\end{eqnarray}
Since $\frac 1 {r_i} = 1 + \frac {t_i}{|r'|}$ and  $r_i = |r_i'|/( 1 + |r_i'|)$, we deduce that
\begin{eqnarray}
   \log N_r(X_1 + \cdots + X_k)
   \geq
   \frac{|r'| \log r }{r} + \sum_{i=1}^k |r'| \left(1 + \frac{t_i}{|r'|} \right) \log  \left(1 + \frac{t_i}{|r'|} \right) +  \sum_{i=1}^k t_i \log \frac{N_r(X_i)}{t_i}.
\end{eqnarray}
It follows that
\begin{equation} N_r(X_1 + \cdots + X_k) \geq c(r,k), \end{equation}
with
\begin{equation}
    c(r,k) \triangleq \inf_\lambda \sup_t  \left( \exp \left\{ \frac{ |r'| \log r}{r} + \sum_{i=1}^k  |r'| \left(1 + \frac{t_i}{|r'|} \right) \log \left(1 + \frac{t_i}{|r'|} \right)+  \sum_{i=1}^k t_i \log \frac{ \lambda_i}{t_i} \right\} \right),
\end{equation}
where the infimum runs over all $\lambda = (\lambda_1 , \dots, \lambda_k)$ such that $\lambda_i \geq 0$ and $\sum_{i=1}^k \lambda_i = 1$, and the supremum runs over all $t =(t_1, \dots, t_n)$ such that $t_i \geq 0$ and $\sum_{i=1}^k t_i =1$.
For a fixed $\lambda$, we can always choose $t = \lambda$,  and thus
\begin{eqnarray}\label{eq:inf}
    c(r,k) 
        &\geq& 
            \inf_t \exp \left\{ \frac{ |r'| \log r}{r} + \sum_{i=1}^k  |r'| \left( 1 + \frac{t_i}{|r'|} \right)\log \left(1 + \frac{t_i}{|r'|} \right) \right\}.
\end{eqnarray}
Due to the convexity of the function $G(u) \triangleq u\log(u)$, $u>0$, we have
\begin{eqnarray}\label{eq:infimum-1}
G\left( 1 + \frac{t_i}{|r'|} \right) \geq G\left( 1 + \frac{1}{k|r'|} \right) + G'\left( 1 + \frac{1}{k|r'|} \right) \left( \frac{t_i}{|r'|} - \frac{1}{k|r'|} \right).
\end{eqnarray}
Using the fact that $\sum_{i=1}^k t_i = 1$, inequality \eqref{eq:infimum-1} yields
\begin{eqnarray}\label{eq:infimum-2}
|r'| \sum_{i=1}^k \left( 1 + \frac{t_i}{|r'|} \right) \log \left(1 + \frac{t_i}{|r'|} \right) \geq \left( k|r'| + 1 \right) \log \left(1 + \frac{1}{k|r'|} \right).
\end{eqnarray}
Since there is equality in \eqref{eq:infimum-2} when $t_i = \frac{1}{k}$, $i = 1, \dots, k$, we deduce that the infimum in \eqref{eq:inf} is attained at $t_i = \frac{1}{k}$, $i = 1, \dots, k$. As a consequence, we have
\begin{eqnarray} \label{eq: c(r,k) bound}
c(r,k) \geq r^{\frac 1 {1-r}} \left( 1 + \frac 1 {k |r'|} \right)^{1 + k |r'|}.
\end{eqnarray}
\end{proof}

\begin{proposition}\label{sharp-k-tuples}
    The largest constant $c_{opt}(r)$ satisfying
    \begin{eqnarray}
        N_r(X_1 + \cdots + X_k) \geq c_{opt}(r) \sum_{i=1}^k N_r(X_i)
    \end{eqnarray}
    for any $k$-tuples of independent log-concave random vectors satisfies
    \begin{eqnarray} 
       e r^{\frac 1 {1-r}} \leq c_{opt}(r) \leq \pi r^{\frac 1 {1-r}}.
    \end{eqnarray}
\end{proposition}

\begin{proof}
Note that the function $k \mapsto \left( 1 + \frac 1 {k |r'|} \right)^{1 + k |r'|}$ decreases to $e$ in $k$.  Thus, taking the limit in \eqref{eq: c(r,k) bound} we have
\begin{eqnarray}
   c_{opt}(r) \geq c(r,k) \geq e r^{\frac 1 {1-r}}.
\end{eqnarray}

On the other hand, specializing the inequality
\begin{equation}
    N_r(X_1+ \cdots+ X_k) \geq c_{opt}(r) \sum_{i=1}^k N_r(X_i)
\end{equation}
to the case in which $X_1, \dots, X_k$ are i.i.d., we must have
\begin{equation} \label{eq: CLT bounds on Renyi} \liminf_{k \to \infty} N_r\left( \frac{X_1 + \cdots + X_k}{\sqrt{k}} \right) \geq c_{opt}(r) N_r(X_1). \end{equation}

Notice that if $X_1$ is a centered log-concave random variable of variance $1$, then the  $\frac{X_1 + \cdots + X_k}{\sqrt{k}}$ are also log-concave random variables of variance $1$, converging weakly by the central limit theorem to a standard normal random variable $Z$.  Moreover, letting $f_k$ denote the density of $\frac{X_1+ \cdots + X_k}{\sqrt{k}}$ one may apply the argument of \cite[Theorem 1.1]{BM18} to $r \in (0,1)$ when one has
\begin{eqnarray} \label{eq:tail bounds for log-concave}
   \lim_{T \to \infty} \int_{\{ |x| > T\}} f_k^r(x) dx = 0
\end{eqnarray}
uniformly in $k$, to conclude that
\begin{equation} \label{eq: convergence in Renyi entropy for CLT} \lim_{k \to \infty} N_r\left( \frac{X_1 + \cdots + X_k}{\sqrt{k}} \right) = N_r(Z) = 2 \pi r^{\frac 1 {r-1}}. \end{equation}
Alternatively, one can arrive at \eqref{eq: convergence in Renyi entropy for CLT} by invoking classical local limit theorems \cite{GK68:book,petrov1964local} to obtain pointwise convergence of the densities, and conclude with Lebesgue dominated convergence to interchange the limit. Recall that the class of centered log-concave densities with a fixed variance can be bounded uniformly by a single sub-exponential function $Ce^{-c|x|}$ for universal constants $C,c >0$ depending only on the variance. This gives the existence of all moments, in particular a third moment requisite for the local limit theorem, additionally it gives domination by an integrable function.

Inserting \eqref{eq: convergence in Renyi entropy for CLT} into \eqref{eq: CLT bounds on Renyi}, we see that $c_{opt}(r)$ must satisfy
\begin{equation}
     2 \pi r^{\frac 1 {r-1}} \geq c_{opt}(r) N_r(X_1).
\end{equation}

For $X_1$ having a Laplace distribution of variance $1$, its density is $f(x) = \frac{e^{-\sqrt{2}|x|}}{\sqrt{2}}$ on $(-\infty,\infty)$, so that 
        \begin{equation}
            N_r(X_1) = 2 r^{\frac{2}{r-1}}.
        \end{equation}
We conclude that
\begin{equation}
\pi r^{\frac 1 {1-r}} \geq c_{opt}(r).
\end{equation}
\end{proof}

Proposition \ref{sharp-k-tuples} shows that there does not exist a universal constant $C$ (independent of $r$ and $k$) such that the inequality
\begin{equation}
    N_r(X_1 + \cdots + X_k) \geq C \sum_{i=1}^k N_r(X_i)
\end{equation}
holds.  Note that this is in contrast with the case $r \geq 1$ when $C = \frac 1 e$ suffices.

\section{Concluding Remarks}

We have shown that a R\'enyi EPI does hold for $r \in (0,1)$, at least for log-concave random vectors, for the R\'enyi EPI of the form \eqref{eq:Bobkov Chistyakov REPI}, as well as for the R\'enyi entropy power raised to a power $\alpha$ as in \eqref{eq: alpha definition}. Let us comment on the sharpness of the $\alpha$ derived, and contrast this behavior with that of the constant derived for uniform distributions $\beta$ from \eqref{eq: definition of beta}.

Due to Madiman and Wang \cite{WM14} the R\'enyi entropy of independent sums decreases on spherically symmetric decreasing rearrangement. Let us recall a few definitions. For a measurable set $A \subset \R^n$, denote by $A^*$ the open origin symmetric Euclidean ball satisfying $\Vol(A) = \Vol(A^*)$.  For a non-negative measurable function $f$, define its symmetric decreasing rearrangement by 
\begin{equation}
	f^*(x) = \int_{0}^{	\infty} \mathbbm{1}_{\{f > t\}^*}(x) dt.
\end{equation}
\begin{theorem}[\cite{WM14}]\label{th:rearrange}
    If $f_i$ are probability density functions and $f_i^*$ denote their spherically symmetric decreasing rearrangements, then
    \begin{equation}
        N_r(X_1 + \cdots + X_k) \geq N_r(X_1^* + \cdots + X_k^*)
    \end{equation}
    for any $r \in [0,\infty]$, where $X_i$ has density $f_i$, and $X_i^*$ has density $f_i^*$, $i = 1, \dots, k$.
\end{theorem}

It follows that to prove inequality \eqref{eq:Bobkov Marsiglietti REPI} it suffices to consider $X$ and $Y$ possessing spherically symmetric decreasing densities. Indeed, using Theorem \ref{th:rearrange} we would have 
\begin{equation}
N_r^{\alpha}(X+Y) \geq N_r^{\alpha}(X^* + Y^*) \geq N_r^{\alpha}(X^*) + N_r^{\alpha}(Y^*) = N_r^{\alpha}(X) + N_r^{\alpha}(Y),
\end{equation}
where the last equality comes from the equimeasurability of a density and its rearrangement. The same argument applies to inequality \eqref{eq:Bobkov Chistyakov REPI}. Motivated by this fact the authors replaced the exponential distribution in the example above with its spherically symmetric rearrangement, the Laplace distribution, to yield a tighter lower bound in an announcement of this work \cite{MarsigliettiMelbourne:isit}. Additionally, since spherically symmetric rearrangement is stable on the class of log-concave random vectors (see \cite[Corollary 5.2]{melbourne2018rearrangement}), one can reduce to random vectors with spherically symmetric decreasing densities, even under the log-concave restriction taken in this work.

\section{Acknowledgements}
The authors thank Mokshay Madiman for valuable comments and the explanation of the R\'enyi comparison results used in this work. The authors are also indebted to the anonymous reviewers whose suggestions greatly improved the paper, leading in particular to the inclusion of Theorem \ref{thm: bobkov chistyakov Repi for small r} and Proposition \ref{sharp-k-tuples}.

\appendix
\section{Proof of Lemma \ref{reverse}}

\begin{theorem}\normalfont{(\cite[Theorem 2.9]{FMW16})}\label{th:log}

For a log-concave function $f$ on $\mathbb{R}^n$, the map
	\begin{equation}
		\varphi(t) = t^n \int_{\mathbb{R}^n} f^t, \qquad t>0,
	\end{equation}
	is log-concave as well.

\end{theorem}

\begin{proof}[Proof of Lemma \ref{reverse}]
	The proof is a straightforward consequence of Theorem \ref{th:log}. What remains is an algebraic computation.  When $X$ has density $f$, one has $\varphi(1) = 1$.  Write $1,p,q$ in convex combination, and unwind the implication of $\varphi$ being log-concave.  We will show the result in the case that we need $0<p<q<1$, the other arguments are similar.  In this case, $\lambda p + (1-\lambda) 1 = q$ for $\lambda = \frac{1-q}{1-p} \in (0,1)$.  By log-concavity,
	\begin{equation}
		\varphi(q) \geq \varphi^\lambda(p) \varphi^{1-\lambda}(1),
	\end{equation}
	which is
	\begin{equation}
		q^n \int f^q \geq \left( p^n \int f^p \right)^{\frac{1-q}{1-p}}.
	\end{equation}
	Since $1-q > 0$ raising both sides to the power $2/n(1-q)$ preserves the inequality, and we have
	\begin{equation}
		q^{2/(1-q)} N_q(X) \geq p^{2/(1-p)} N_p(X).
	\end{equation}
	which implies our result.
\end{proof}

\section{Positivity of $A(p,q,r)$}
By \eqref{eq: x y description of A}, it suffices to show that
\begin{equation}
    W(x,y) = \log \left( \frac{(x+y)^{x+y} (x+1)^{x+1} (y+1)^{y+1}}{x^x y^y (x+y+1)^{(x+y+1)}} \right) > 0,
\end{equation}
for $x, y >0$.  First observe that for $y > 0$,
\begin{equation}
    \lim_{x \to 0} W(x,y) = 0.
\end{equation}
Computing,
\begin{equation}
    \frac{\partial}{\partial x} W(x,y) = \log \left( \frac{ (x+y)(x+1)}{(x+y+1)x} \right),
\end{equation}
which is always greater than $0$, since
\begin{equation}
    (x+y)(x+1) > (x+y+1)x
\end{equation}
reduces to $y >0$.  Thus $W(x,y) > W(0,y) = 0$ for $x,y>0$, and our result follows.
\bibliographystyle{plain}
\bibliography{bibibi}

\vspace*{1cm}

\noindent Arnaud Marsiglietti \\
Department of Mathematics \\
University of Florida \\
Gainesville, FL 32611 \\
E-mail address: a.marsiglietti@ufl.edu

\vspace{0.8cm}

\noindent James Melbourne \\
Electrical and Computer Engineering \\
University of Minnesota \\
Minneapolis, MN 55455, USA \\
E-mail address: melbo013@umn.edu

\end{document}